\documentclass{article}
\usepackage{amsmath,amsfonts,amssymb,amsthm,bbm,setspace}
\usepackage{algorithm} 
\usepackage{algpseudocode} 
\usepackage{mathtools}
\usepackage{commath}
\usepackage{graphicx}
\usepackage{xcolor}
\usepackage{setspace}
\usepackage{hyperref}
\usepackage[margin = 1 in]{geometry}
\newcommand\mc{\mathcal}
\newtheorem{theorem}{Theorem}[subsection]

\newtheorem{lemma}{Lemma}[subsection]
\newtheorem{example}{Example}[subsection]
\title{Noisy Nonadaptive Group Testing with Binary Splitting: New Test Design and Improvement on Price-Scarlett-Tan's Scheme}
\author{Xiaxin Li \and Arya Mazumdar
\thanks{The authors are with the Hal{\i}c{\i}o\u{g}lu Data Science Institute, University of California, San Diego. emails: \texttt{\{xil095, arya\}@ucsd.edu}.
}}
\date{}
 
\begin{document}
\maketitle
\begin{abstract}
In Group Testing, the objective is to identify $K$ defective items out of $N$, $K\ll N$, by testing pools of items together and using the least amount of tests possible. Recently, a fast decoding method based on {\em binary splitting} (Price and Scarlett, 2020) has been proposed that simultaneously achieve optimal number of tests and decoding complexity for Non-Adaptive Probabilistic Group Testing (NAPGT). However, the method works only when the test results are noiseless. In this paper, we further study the binary splitting method and propose (1) A NAPGT scheme that generalizes the original binary splitting method from the noiseless case into tests with $\rho$ proportion of false positives (the $\rho$-False Positive Channel), where $\rho$ is a constant, with asymptotically-optimal number of tests and decoding complexity, i.e. $\mc{O}(K\log N)$, and (2) A NAPGT scheme in the presence of both false positives and false negatives in test outcomes, improving and generalizing the work of Price, Scarlett and Tan~\cite{price2023fast} in two ways: First, under $\rho$-proportion of test results flipped ($\rho$-Binary Symmetric Channel) and within the general sublinear regime $K=\Theta(N^\alpha)$ where $0<\alpha<1$, our algorithm has a decoding complexity of $\mc{O}(\epsilon^{-2}K^{1+\epsilon})$ where $\epsilon>0$ is a constant parameter. 
% Furthermore, given $K=\Omega(N^\beta)$ for some $0<\beta<1$, this can be  improved to $\mc{O}(\epsilon^{-1}K^{1+\epsilon})$. 
Second, when the false negative flipping probability $\rho'$ satisfies $\rho'=\mc{O}(K^{-\epsilon})$ and the false positive flipping probability $\rho$ is a constant, we can simultaneously achieve
$\mc{O}(\epsilon^{-1}K\log N)$ for both the number of tests and the decoding complexity. It remains open to achieve these optimals under the general BSC. 
\end{abstract}
\section{Introduction}
Originally proposed by Dorfman \cite{dorfman1943detection} during World War II to detect syphilis among soldiers, \emph{Group Testing (GT)} refers to the process of identifying a certain number of defective items out of a set of given items, by performing tests on pools of items and finding the defective items by the results of the tests. A test result is positive if and only if the tested pool contains at least one defective item.   Compared to individual testing, where each pool contains only one item, GT is significantly more efficient, in terms of the number of tests and the decoding complexity. Throughout the years, the GT techniques have been successfully applied in various fields, such as detecting COVID-19 \cite{hogan2020sample,verdun2021group,yelin2020evaluation,brust2023effective}, compressed sensing \cite{matsumoto2023improved,matsumoto2024robust,gilbert2007one,gilbert2008group}, DNA sequencing \cite{hwang2006pooling,cao2016combinatorial}, wireless communication \cite{berger1984random,luo2008neighbor}, etc. 

Depending on the applications, group testing can be performed either adaptively, where the subsequent tests depend on the outcome of previous tests, or nonadaptively, where all tests are performed simultaneously. In this paper, we restrict ourselves to the nonadaptive version of group testing. Also, our pooling designs are guaranteed to exactly recover defectives with high probability; this is also known as probabilistic group testing or for-each recovery.
 We exclusively focus on \emph{non-adaptive probabilistic group testing (NAPGT)} in this paper. 

In recent years, there has been a significant number of publications dealing with NAPGT \cite{aldridge2019group,atia2012boolean,cheraghchi2020combinatorial,cai2017efficient,price2023fast,price2020fast,lee2019saffron,flodin2021probabilistic,mazumdar2015nonadaptive,wang2023quickly}. Under the noiseless settings, where the test results are perfectly reliable,  Price and Scarlett \cite{price2020fast} and Cheragchi and Nakos \cite{cheraghchi2020combinatorial} simultaneously proposed the \emph{fast binary splitting algorithm} that achieves the theoretically optimal asymptotic bound $\mc{O}(K\log N)$, where \emph{$N$ is the total number of items} and \emph{$K$ is the upper bound of the number of defective items}, for both the number of tests and the decoding complexity, a significant result showing that the asymptotic-optimality of number of tests and decoding complexity can be simultaneously achieved. Very recently, Wang, Gabrys, and Guruswami \cite{wang2023quickly} generalized the result of \cite{price2020fast} so that the constant factor of the number of tests is smaller: from the constant 16 to a constant arbitrarily close to $\log(2)^{-2}$. 

However, under the noisy settings, where the test results may flip from negative (0) to positive (1) with probability $\rho$, called a \emph {$\rho$-False Positive Channel ($\rho$-FPC)}, or the other way, from 1 to 0 with probability $\rho'$, called a \emph{$\rho'$-False Negative Channel ($\rho'$-FNC)}, or a combination of both channels, i.e. a \emph{($\rho,\rho'$)-Binary Asymmetric Channel (($\rho,\rho'$)-BAC)}, it remains an open question to come up with a GT scheme that achieves the asymptotic-optimal bounds for both the number of tests and decoding complexity. One of the best results to-date is the recently proposed Gacha scheme \cite{guruswami2023noise}, which works with the general \emph{$\rho$-Binary Symmetric Channel ($\rho$-BSC)} where the flipping probability of the test results from either 1 to 0 or from 0 to 1, is the same constant $0<\rho<\frac{1}{2}$. Gacha achieves the optimal bound $\mc{O}_\rho (K\log N)$ for the number of tests where $\mc{O}_\rho$ hides a constant that depends on the BSC flipping probability $\rho$, while getting $\mc{O}_\rho(K(\log N)^{3/2})$ for the decoding complexity, within the sparsity regime where $\log K=\mc{O}((\log N)^{1-\epsilon})$ for any chosen $\epsilon>0$. Another important result by Price, Scarlett, and Tan \cite{price2023fast} also works with the general $\rho$-BSC and they achieve the bound $\mc{O}(\epsilon^{-1}K\log N)$ for the number of tests where $\epsilon>0$ is a constant parameter while getting $\mc{O}(\epsilon^{-1}(K\log N)^{1+\epsilon})$ for the decoding complexity.

In this paper, inspired by ideas from various existing NAPGT schemes \cite{guruswami2023noise,wang2023quickly,cai2017efficient,price2023fast,atia2012boolean,price2020fast,cheraghchi2020combinatorial}, we propose two NAPGT schemes. Our first scheme assumes a $\rho$-FPC, in which we propose a new testing design, achieving $\mc{O}(K\log N)$ for both the number of tests and the decoding complexity. Our second scheme considers the sublinear sparse regime $K=\Theta(N^\alpha)$ where $0<\alpha<1$, improving and generalizing the fast binary splitting algorithm by Price-Scarlett-Tan \cite{price2023fast} in two ways. First, under the $\rho$-BSC, we improve the decoding complexity from $\mc{O}(\epsilon^{-1}(K\log N)^{1+\epsilon})$, to $\mc{O}(\epsilon^{-2}K^{1+\epsilon})$ where $\epsilon>0$ is a constant parameter that can be chosen arbitrarily small. Second, under the $(\rho,\rho')$-BAC where $\rho'=\mc{O}(K^{-\epsilon})$, and $\rho$ is a constant with a mild constraint, we can achieve $\mc{O}(\epsilon^{-1}K\log N)$ for both the number of tests and the decoding complexity. 

The following table summarizes some established results in the NAPGT under the $\rho$-BSC, within the sublinear regime $K=\Theta(N^\alpha)$, $0<\alpha<1$, with the exception of Gacha, which only works for the regime $\log K=(\log N)^{1-\epsilon}$. 
\\\\
\begin{tabular}{|c||c|c|c|c|}
    \hline
    Name &  GROTESQUE\cite{cai2017efficient}& Price-Scarlett-Tan\cite{price2023fast} & Gacha\cite{guruswami2023noise} & This paper \\
    \hline
    \# Tests & $\mc{O}(K\log K\log N)$ & $\mc{O}(\epsilon^{-1}K\log N)$ & $\mc{O}_\rho(K\log N)$ & $\mc{O}(\epsilon^{-1}K\log N)$ \\ 
    \hline
    \# Decoding & $\mc{O}(K\log N)$ & $\mc{O}(\epsilon^{-1}(K\log N)^{1+\epsilon})$ & $\mc{O}_\rho(K(\log N)^{3/2})$ & $\mc{O}(\epsilon^{-1}K^{1+\epsilon})$ \\
    \hline
\end{tabular}
\\

\textbf{The paper will be organized as follows.} Section~\ref{sec:Prelim} defines the problem and summarizes some notations that will be used throughout the paper. In addition, it also provides a background on relevant methodologies that will be used in our scheme. Section~\ref{sec:noiselessneon} presents our GT scheme in the noiseless setting and concludes with Theorem~\ref{thm:noiselessneon}. Section~\ref{sec:fpcneon} generalizes Section~\ref{sec:noiselessneon} into the $\rho$-FPC setting and concludes with Theorem~\ref{thm:fpcneon}. Section~\ref{sec:BSCneon} presents our GT scheme under the $\rho$-BSC setting and concludes with Theorem~\ref{thm:maintheoremBSC}. Section~\ref{sec:BACneon} presents our GT scheme under the $(\rho,\rho')$-BAC, and concludes with Theorem~\ref{thm:maintheoremBAC}. Constructions of test matrices and decoding algorithms for NEON are summarized in Algorithms~\ref{alg:noiseless-test}, \ref{alg:noiseless-decoding}, \ref{alg:noisy-test}, \ref{alg:noisy-decoding}.

\section{Preliminaries}\label{sec:Prelim}
\subsection{Group Testing}
Let $N$ denote the \emph{total number of items}, let $K$ denote the \emph{upper bound of the number of defective items}, and let $x$ be an $N$-dimensional binary vector with a Hamming weight less or equal to $K$. 
Recall from the introduction that we will focus on the NAPGT, so we assume each possible instance of the vector $x$ is chosen with the same probability. 

Let $M$ be the \emph{number of tests}, and let $A$ be a $M\times N$ binary matrix which we call the \emph{test matrix}. 
Let $y:=A\circ x$ be the \emph{test result vector}, wherein the ``$\circ$'' operation replaces the usual addition and multiplication to logical-or and logical-and. More specifically, for $i\in\{1,2,..,M\}$, we have 
$$y_i:=\bigvee_{j=1}^N (A_{ij}\wedge x_j).$$

Our goal is to design the test matrix $A$ such that for any input vector $x$ chosen with uniform probability, which corresponds to a defective subset $S$ of items of cardinality less than or equal to $K$, this scheme will utilize the test result vector $y$ and generate an estimate $\widehat{S}$ such that $\mathbb{P}(\widehat{S}\neq S)<\epsilon_{N,K}$ where $\epsilon_{N,K}\rightarrow 0$ as $N, K\rightarrow \infty$. 

Let $D$ denote the \emph{decoding complexity}, we aim for both $M$ and $D$ to be of scale $\mc{O}(K\log N)$.
We will name our schemes in this paper as \textbf{NEON}, which stands for Noise-resilient, Efficient and near-Optimal testiNg.  

Denote $\alpha$ as the \emph{sparse regime parameter}, i.e. $K=\Theta(N^\alpha)$, which will be a requirement in Section~\ref{sec:BSCneon} and Section~\ref{sec:BACneon}. 
Denote $\rho$ as the \emph{False Positive flipping probability}, and $\rho'$ as the \emph{False Negative flipping probability}.
Throughout the paper, we will use $\log(\cdot)$ to denote the natural logarithm. 

In Noiseless and $\rho$-FPC NEON, we will use local matrices that aim to recover only $ \lfloor\log K\rfloor$ defectives out of $N$ items, so we denote $k:=\log K$.

The following table summarizes the notations that will be used throughout the paper: 

\begin{center}
\begin{tabular}{|c|c|}
\hline
   $N$ & Number of items \\
\hline
   $K$ & Upper bound of Number of defective items \\
\hline
   $M$ & Number of tests \\
\hline
   $D$ & Number of operations in the decoding algorithm \\
\hline
   $\alpha$ & The sparse regime parameter. i.e. $K=\Theta(N^\alpha)$, used in Section~\ref{sec:BSCneon} and~\ref{sec:BACneon}\\
\hline
   $\rho$ & False Positive flipping probability \\
\hline
   $\rho'$ & False Negative flipping probability \\
\hline
   $k$ & $\log K$ \\
\hline
\end{tabular}
\end{center}

\subsection{Background}
In this section, we give a brief description of a few existing group testing algorithms that we will later rely on for our constructions.

\subsubsection{Gacha scheme}\label{sec:gacha}
 In the Gacha scheme \cite{guruswami2023noise}, a test matrix is constructed such that each column can be viewed as a codeword of a certain block length. 
In other words, the test matrix can be viewed as a \emph{symbol matrix}. The scheme is defined by three parameters: $\mc{L}$, $\mc{S}$ and $C$.

$\mc{L}$ denotes the block length, and $\mc{S}$ denotes the number of binary entries to represent each symbol. We have $M=\mc{L}\cdot\mc{S}$. In the case of Gacha, these were $\mc{L}=8K\sqrt{\log_2 N}$ and $\mc{S}=7\sqrt{\log_2 N}$. 
\begin{example}\label{exp:bl}
Take $N=5$ and $\mc{L}=3$. Suppose the original test matrix viewed as a symbol matrix is: 
$$
\begin{bmatrix}
    A&B&C&L&D
    \\
    K&N&O&F&G
    \\
    H&J&I&M&E
\end{bmatrix}
$$

The codewords, from left to right, are $AKH, BNJ,COI,LFM,$ and $DGE$, respectively. 

In this case, we can take $\mc{S}=4$. Using the binary representation of the first fifteen letters, i.e. $A\rightarrow 0001$, $B\rightarrow 0010$, ..., $O\rightarrow 1111$. The original test matrix will be the following:
$$
\begin{bmatrix}
    0&0&0&1&0
    \\[-3pt]
    0&0&0&1&1
    \\[-3pt]
    0&1&1&0&0
    \\[-3pt]
    1&0&1&0&0
    \\[3pt]
    1&1&1&0&0
    \\[-3pt]
    0&1&1&1&1
    \\[-3pt]
    1&1&1&1&1
    \\[-3pt]
    1&0&1&0&1
    \\[3pt]
    1&1&1&1&0
    \\[-3pt]
    0&0&0&1&1
    \\[-3pt]
    0&1&0&0&0
    \\[-3pt]
    0&0&1&1&1
\end{bmatrix}.
$$
\end{example}

The third parameter in Gacha, the ``circling parameter'' $C$ (which equals $4\sqrt{\log_2 N}$ in Gacha), refers to the number of symbols in each column of the original test matrix that are kept, while the rest of the symbols are zeroed out. The resulting matrix would be the (final) \emph{test matrix for Gacha}. 

\begin{example}\label{exp:bl2}
(Continuation of Example~\ref{exp:bl}) Suppose $C=1$ in this example, and suppose we circle the 1st, 3rd, 3rd, 2nd, and 1st symbol in columns 1 to 5, respectively, after the circling and zeroing process, the original test matrix will become the following test matrix: 
$$
\begin{bmatrix}
    A&0&0&0&D
    \\
    0&0&0&F&0
    \\
    0&J&I&0&0
\end{bmatrix}
$$
\end{example}

In Noiseless and $\rho$-FPC NEON, we will regard our test matrix $A$ as a vertical concatenation of $\mc{L}$ blocks, where each block has $\mc{S}$ rows and each column within a block is regarded as a symbol. Then we will circle $C$ symbols in each column while zeroing other symbols. 

\subsubsection{Fast binary splitting}\label{sec:fbs}
In the fast binary splitting algorithm~\cite{price2020fast,cheraghchi2020combinatorial,wang2023quickly,price2023fast}, a test matrix is constructed based on a binary tree of depth $\log_2 N$, where each node represents a group of items. Note that we can associate each item with a binary string of length $\log_2 N$, from item $[00...0]$ to item $[11...1]$. 

The top node of this binary tree represents the whole set of items, and its two children will each represent all items that begin with $0$ and $1$. The two children of the node $0$ will represent items that begin with $00$ and $01$, and the two children of the node $1$ will represent items that begin with $10$ and $11$. This representation continues to every node of the tree, so for the last level of the tree, each node represents a single item. For convenience and without loss of generality, we suppose $\log_2(K)$ is an integer. In the following discussion, we will use the constructions in~\cite{wang2023quickly}, which consists of the following two phases: 

The first phase, named ``grow-and-prune'', contains a total of $16K\log_2(N/K)$ tests. Starting at level $\log_2(K)+1$ of the binary tree to the final level, each level consists of $16K$ tests where each node of the level is assigned to one of these tests uniformly at random. When we say a node is assigned to a test, we mean that all the items represented by this node will be included in this test. Regarding the constant 16 here, ~\cite{wang2023quickly} showed that it can be modified to $\log(2-4\epsilon)^{-2}$ for any $\epsilon>0$. We will refer to this constant as $\zeta$ in our NEON scheme. 

The second phase, named ``leaf-trimming'', contains a total of $16K\log_2(K)$ tests with labels $(l_1,l_2)$ such that $1\leq l_1\leq \log_2(K)$ and $1\leq l_2\leq 16K$. Each item is assigned to tests $(l,t)$ for every $1\leq l\leq \log_2(K)$ and for any fixed $l$, the $t$ is chosen uniformly random. The ``leaf-trimming'' phase is an important step, as this will significantly contribute to the elimination of false positive errors from the first phase. However, in Noiseless and $\rho$-FPC NEON, as we will see, we will skip this phase because of our new testing design.

Regarding the decoding algorithm, the decoder starts at level $\log_2(K)+1$ of the binary tree and examines the test results of each node at this level. If a node is tested positive, then the decoder continues to examine the test results of the children of this node, until it reaches the final level. As we expect $K$ defectives in total, with high probability, the decoder will terminate within $\mc{O}(K\log_2 N)$ operations. 

In the analysis of this algorithm, we want to keep track of the number of positive nodes at level $l$ during the decoding process. In~\cite{wang2023quickly}, the random variable $N_l$ is used to denote this number minus $K$. To find $N_l$,~\cite{wang2023quickly} defines a probability generating function $F_l(q)$ by $F_{\log_2(K)}(q)=1$ and $F_{l+1}(q)=F_l(\frac{15}{16}+\frac{q^2}{16})\cdot q^K$. It has been shown that $F_l'(1)$ equals to the mean of $N_l$, and $N_n<cK$ with high probability for a reasonable constant $c$. For example,~\cite{wang2023quickly} proved that $\mathbb{P}(N_n\geq 5K)<e^{-K}$. In some of our proofs, we will use the notations above. 

\subsubsection{Sub-tree decoding}\label{sec:std} 
We will use the idea of sub-tree decoding from~\cite{price2023fast}. 

Recall from the previous section that the decoding algorithm of the original fast binary splitting algorithm only examines the result of each node. In the noiseless case or the False Positive-only setting this will work properly. However, under the False Negative setting where the test results may flip from $1$ to $0$, such a decoding algorithm will not work, as the decoder will stop at a False Negative node, which can hide many truly positive items. 

To remedy this issue,~\cite{price2023fast} proposed the idea of sub-tree decoding: To justify whether a node is positive or negative, instead of only looking at the test result of this node, the decoder looks down a subtree of height $r$ of this node and studies the results of all $2^r-2$ nodes in this subtree. If there exists a branch of this subtree that more than half of its nodes are positive, then the original node is claimed to be positive. 

This allows the correction of False Negative nodes and by analysis, the decoder in~\cite{price2023fast} can find all defective items with high probability. 

In~\cite{price2023fast}, the height $r$ is chosen to be dependent on the $N$ and $K$. Specifically, $r=\mc{O}(\log K+\log\log N)$, resulting in a sub-optimal decoding complexity. In $\rho$-BSC NEON, we will instead choose $r=\mc{O}(\log K)$ and apply a slightly different decoding method, resulting in an improvement of a factor up to $\mc{O}((\log N)^{1+\epsilon})$ in decoding complexity. 
On the other hand, in $(\rho,\rho')$-BAC NEON where $\rho'=\mc{O}(K^{-\epsilon})$, we can choose this $r$ to be a constant.

\section{Noiseless NEON scheme}\label{sec:noiselessneon}
Before getting into the $\rho$-FPC NEON scheme, we suppose there is no noise and present the Noiseless NEON scheme. 
In the following section, we start with the construction of one concrete example of the NEON scheme, named as \emph{BasicNEON}. After this, we will generalize BasicNEON to the general Noiseless NEON. 
\subsection{BasicNEON: A first example}\label{sec:firstexample}
To construct the BasicNEON test matrix $A$, we start with a simplified fast binary splitting test matrix $B$, which can be regarded as a local component of $A$. 

By the word simplified, we mean that the matrix $B$ will be constructed without the leaf trimming phase (see Section~\ref{sec:fbs}), and only aims to recover up to $k=\log K$ defectives. 

By the result in~\cite{wang2023quickly},  with $\mc{O}(k\log N)$ rows and $N$ columns, such matrix $B$ exists to recover a superset of size $<5k$ of up to $k$ defectives with high probability. 

% Note that globally we aim to exactly recover up to $K$ defectives, with high probability with the test matrix $A$. One reason for using $k$ instead of $K$ here is that the matrix $B$ is only a local component (i.e. a block) of the whole matrix $A$, which is only required to recover $O(k)$ defectives rather than $K$ defectives. 

The reason for skipping the leaf trimming phase in choosing the matrix $B$ is that the failure probability will be much higher if we include the leaf trimming phase. On the other hand, excluding the leaf trimming phase will lead to superset recovery rather than the exact recovery. This issue will be solved by our decoding algorithm. 

Having the matrix $B$, we construct the test matrix $A$ using copies of $B$ as building blocks. Recall from Section~\ref{sec:gacha}, our matrix $A$ will be a vertical concatenation of $\mc{L}$ blocks, where each block has $\mc{S}$ rows and each column of the block is regarded as a symbol. In BasicNEON, we take the blocks to be the matrix $B$, and we concatenate $\mc{L}=\frac{\lambda K}{k}$ copies of $B$ vertically, each with $\mc{S}=\zeta kn$ rows, where $\lambda$ and $\zeta$ are constants. Label the blocks from $B_1$ to $B_{\mc{L}}$. Now, we get our initial matrix $A'$.  

After getting $A'$, we apply the circling process: We circle $C$ symbols in each column of $A'$ uniformly random, where $C$ is a constant. In BasicNEON, we take $C:=10$. Then we zero out uncircled symbols to get our test matrix $A$. 

The figure below illustrates our construction: 
\begin{center}
\includegraphics[height=6cm,width=18cm]{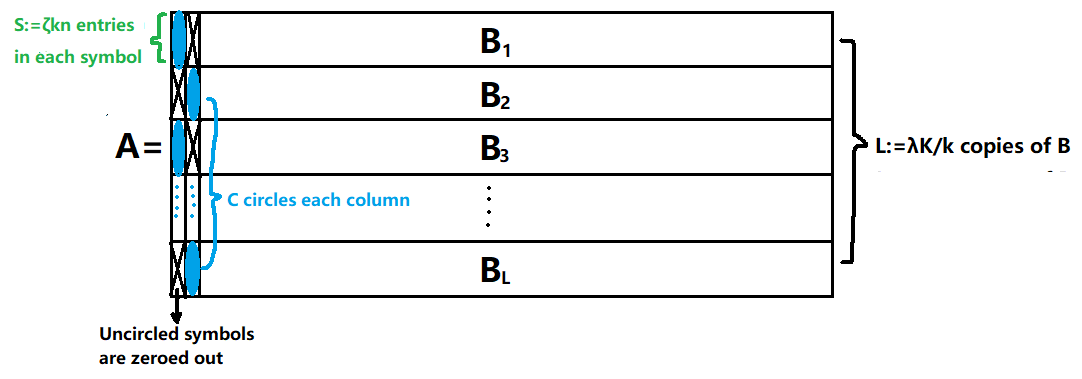}
\end{center}

\emph{For the decoding process,} we will decode locally and combine the local results to get our global result. 

Specifically, we use the same decoder for each block $B_j$ where $j=1,2,...,\mc{L}$, and we call this \emph{local decoder} (LD). Let $\mc{D}_j$ be the decoding result from the block $B_j$. 

For each item that appeared in at least one of the $\mc{D}_j$s, we count the total number of $\mc{D}_j$s that it is in. In the final stage of the decoding algorithm, we claim that an item is defective only if it is in $C$ of the $\mc{D}_j$s. 

The power of this idea is to filter out the error of one single decoder by other decoders. For example, $\mc{D}_1$ may wrongly claim the $i$-th item to be defective but the $i$-th item is, in fact, non-defective. The undesired event that many other $\mc{D}_j$s will include this item is improbable. 

On the other hand, for a truly defective item, $C$ of the $\mc{D}_j$s will claim it to be defective, as we circle $C$ symbols in each column, and there is no false negatives. 
% The presence of false negative items will potentially impact the reliability of this decoding algorithm, but we will show that we can avoid finding any false negatives with high probability. 

In Section~\ref{sec:initial} and Section~\ref{sec:errorprob}, we will analyze the BasicNEON scheme in detail, and the following example illustrates our decoding algorithm on a global perspective:
\begin{example} In this example, we take $N=18$, $\mc{L}=22$, and $K=5$, where the items 6-10 are defective, denoted by \textcolor{blue}{blue} entries.
Recall that $C=10$, so we expect to have $10$ circles for each column for the defectives. From another perspective, for each fixed defective item, we expect to have $10$ decoders finding it. 
Also, for every non-defective item, we expect them to be wrongly found by very few decoders.

The following diagram is a typical realization of this scenario. A \textcolor{green}{green} entry denotes a defective detected by the local decoder, and an \textcolor{orange}{orange} entry denotes a false positive. 

\begin{center}
\includegraphics[height=9cm,width=13cm]{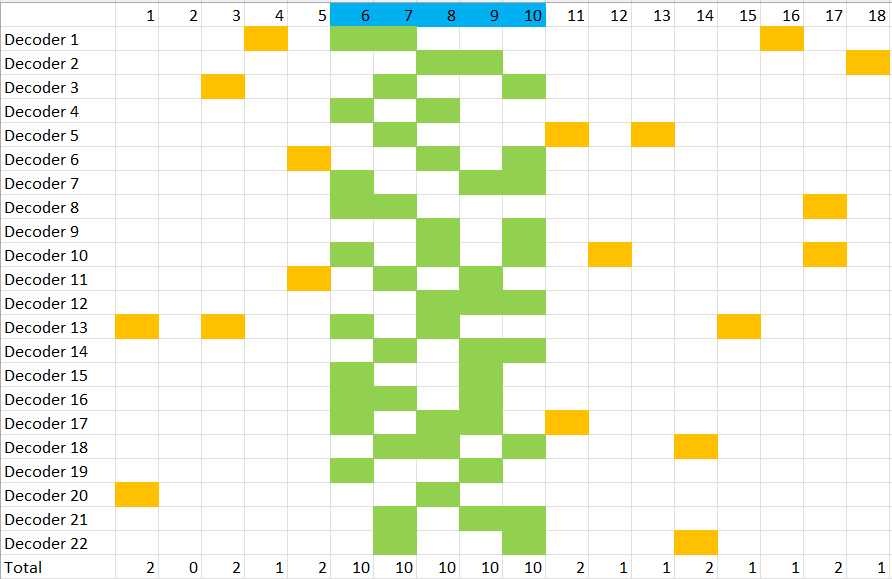}
\end{center}

By our decoding algorithm, an item will be claimed defective if it appears in 10 decoders, and it will be claimed non-defective if it appears in less than 10 decoders. 

In the realization above, we can find the defective items $6-10$ accurately. as all other items appear in $<10$ decoders.
\end{example}

\subsection{Initial analysis of BasicNEON}\label{sec:initial}
Consider the defective part of the test matrix $A$, denote as $A_d$, which has at most $K$ columns.

\emph{First}, we analyze the probability that our scheme is valid. By valid, we mean that each block $B_j$ of $A$ will contain $\leq k$ defectives. In other words, the intersection of every $B_j$ with $A_d$ contains $\leq k$ circles. 

Define $p$ as the probability that one fixed $B_j\cap A_d$ has $>k$ circles. 
% To compute $p$, we need the information on whether these $K$ symbols in the same row are circled or not.

Note that a symbol is circled with probability $\frac{C}{\mc{L}}$. Let $X_i$, $i=1,2,...,K$ be Bernoulli random variables Ber$(\frac{C}{\mc{L}})=$ Ber$(\frac{Ck}{\lambda K})$. 
Let $Y=X_1+...+X_K$, so $Y\sim$ Bin$(K,\frac{C}{\mc{L}})$ and $\mathbb{E}[Y]=\frac{Ck}{\lambda}$.
Also let $\delta:=\frac{\lambda}{C}$, then $\mathbb{E}[Y]=\frac{k}{\delta}$. By Chernoff bound, we have 
$$p=\mathbb{P}(Y>k)=\mathbb{P}(Y>\delta \mathbb{E}[Y])\leq \left(\dfrac{e^{\delta-1}}{\delta^\delta}\right)^{\mathbb{E}[Y]}=\left(\dfrac{e^{\delta-1}}{\delta^\delta}\right)^{\frac{k}{\delta}}<\left(\dfrac{e}{\delta}\right)^k.$$
Let $q:=\left(\dfrac{e^{\delta-1}}{\delta^\delta}\right)^{\mathbb{E}[Y]}$ be the upper bound of $p$.
When $\delta \gg e$, $q$ will be $\ll K^{-1}$. 

For instance, when $\delta:=e^3$,  we have $q<e^{-2k}=K^{-2}\ll K^{-1}$. 

In Section~\ref{sec:errorprob}, we will utilize $p<q<K^{-2}$. 

% Next, we make the following remarks related to the error probability of the NEON scheme, which will be analyzed in more detail in Section 3.3. 

% (i) Note that in NEON, we don't want more than $k$ circles to be in the same row, so we need to avoid the event that there exists a row with more than $k$ circles. 

\emph{Next}, we analyze the probability of the appearance of a false positive. 

For a fixed non-defective item, define its $R$-value as the number of local decoders (LDs) that incorrectly claim it to be defective. Specifically, let $V_j$ for $j=1,2,...,\mc{L}$ be Ber$(\gamma)$ for some $\gamma>0$, the indicator of erroneous detection (i.e. claiming this non-defective item as defective) for the $j$-th LD, then $R:=V_1+...+V_\mc{L}$.

We have $\mathbb{E}[R]=\mathbb{E}[V_1]+...+\mathbb{E}[V_{\mc{L}}]=\mc{L}\cdot \gamma$. 
We need $\gamma$ to be as small as possible so that the probability $\mathbb{P}(R>C)$ is very small, to the extent that there does not exist any non-defective item with its corresponding $R$ value being greater than $C$. Recall that in the decoding algorithm of NEON, for each item, defective or non-defective, we will count the number of LDs claiming it to be defective. If this item is defective, we will have $C$ LDs claiming it to be defective. If this item is non-defective, then we need the number of LDs claiming it to be defective to be less than the threshold value $C$, which is equivalent to its $R$-value being less than $C$. 

\begin{lemma}\label{lem:const}
Choose the constant $\zeta:=16$, then $\gamma<\frac{2k}{N}$.   
\end{lemma}  

\begin{proof}
    
Referring to the analysis in Section II.E of \cite{wang2023quickly}, notice the similar meaning of the constant $\zeta$ in this paper and the constant 16 in \cite{wang2023quickly}, and note that $\gamma$ means the probability of a fixed non-defective item being wrongly detected, which is $\leq\frac{N_n}{N}$ when $N_n$ is fixed. We have: 
$$\gamma=\dfrac{\mathbb{E}[N_n]}{N}=\frac{F_n'(1)}{N}<\frac{8k}{7N}<\frac{2k}{N}.$$
The notations $N_n$ and $F_l'(1)$ were mentioned in Section~\ref{sec:fbs}. Note that $F_l'(1)$ is defined recursively by $F_{\log_2 k }'(1)=0$ and $F_{l+1}'(1)=\frac{1}{8}\cdot F_l'(1)+k$, so it is clear that for any $l>0$, $F_l'(1)<\frac{8}{7}k<2k$. 
\end{proof} 

Recall in Section~\ref{sec:firstexample}, in BasicNEON, $C:=10$. \emph{Here, we further impose} $\alpha<0.8$ for BasicNEON. 

In the next lemma, we obtain an upper bound for the undesired event $R>C$:
\begin{lemma}\label{lem:basicneon}
    In BasicNEON, we have $\mathbb{P}(R>C)=\mc{O}(N^{-2})$. 
\end{lemma} 
\begin{proof}
     Starting with $R=V_1+...+V_\mc{L}$, by Chernoff bound, for any $m>0$, we have
$$\mathbb{P}(R>m\mathbb{E}[R])<\left(\frac{e^{m-1}}{m^m}\right)^{\mathbb{E}[R]}.$$
From the expression above, pick $m$ such that $m\mathbb{E}[R]=10=C$ and note that $\alpha<0.8$.

By Lemma~\ref{lem:const}, we have
$$\mathbb{P}(R>m\mathbb{E}[R])<\left(\frac{e^{m-1}}{m^m}\right)^{\frac{10}{m}}<\left(\frac{e^{m}}{m^m}\right)^{\frac{10}{m}}=\left( \frac{e}{m}\right)^{10}=\left(\frac{e}{10}\cdot\mathbb{E}[R]\right)^{10}$$
$$=\left(\frac{e}{10}\cdot(\mc{L}\cdot \gamma)\right)^{10}<\left(\frac{e}{10}\cdot \mc{L}\cdot\frac{2k}{N}\right)^{10}=\left(\frac{e}{10}\cdot \frac{\lambda K}{k}\cdot\frac{2k}{N}\right)^{10}=\mc{O}(N^{-10(1-\alpha)})\ll N^{-2}.$$
\end{proof}

\subsection{Error probability analysis of BasicNEON}\label{sec:errorprob}
In this section, we summarize and analyze all possible error events of the BasicNEON scheme and prove that the overall error probability approaches zero when $K\rightarrow\infty$ and $N\rightarrow\infty$. 
\\\\
\textbf{Error 1: There exists a row with more than $k$ circles}

Recall that $p$ is the probability that one fixed $B_j\cap A_d$ has $>k$ circles. 
Thus, the desired event  $B_j\cap A_d$ having $\leq k$ circles happens with probability $1-p$. 

There are $\mc{L}$ $B_j$s, so the probability of all $B_j$s having $\leq k$ circles intersecting with $A_d$ is $\mc{O}((1-p)^\mc{L})$. The big $O$ notation is used here because the events ``$B_j$ has $>k$ circles'' are not strictly independent. 

However, without loss of generality, we suppose the probability above is exactly $(1-p)^{\mc{L}}$. Then, the Error 1 probability is $p_1:=1-(1-p)^\mc{L}$. 

As in Section~\ref{sec:initial}, choosing $\delta:=e^3$, and note that $p<K^{-2}$, we have 
$$(1-p)^\mc{L}=(1-p)^{\frac{\lambda K}{k}}=\mc{O}(1-\frac{\lambda}{K\log K}).$$
$$\Rightarrow p_1=\mc{O}(\frac{\lambda}{K\log K})=\mc{O}(K^{-1}).$$
Note that by changing the constant $\delta$ to be a larger power of base $e$, we can make the error probability $p_1$ to be $\mc{O}(K^{-b})$ for any $b>0$. This will be analyzed in Section~\ref{sec:generalized}. 
\\\\
\textbf{Error 2: There exists a non-defective item (column) with $R>C$.}

Recall that the undesired event $R>C$ happens with probability $<N^{-2}$, by the analysis in Section~\ref{sec:initial}. Now there are at most $N$ items, so that every non-defective item has a $R$-value less than $C$ happens with probability at least $(1-N^{-2})^N = 1-\mc{O}(N^{-1})$.

Hence, the Error 2 probability is 
\[p_2=1-(1-N^{-2})^N=\mc{O}(N^{-1}).\] 
% \textbf{Error 3: There exists a False Negative item}

% This will not be an issue for the Noiseless NEON and general FPC NEON and will be solved by an alternative method (False Negative Correction) in the discussion of BAC and BSC NEONs. 

In summary, we get the following theorem: 
\begin{theorem}\label{thm:basicneon} (BasicNEON) For the regime $K=o(N^{0.8})$, there exists a noiseless, exact-recovery probabilistic GT scheme with $\lambda\zeta K\log_2 N=(C*\delta)*\zeta*Kn=10*e^3*16*Kn=160e^3K\log N$ tests and decoding complexity $\mc{O}(K\log N)$, with error probability $<p_1+p_2=\mc{O}(K^{-1}+N^{-1})$.
\end{theorem}
% \begin{proof} The proof of this theorem can be seen as the summary of all previous contents. First, recall that we choose the local blocks $B_j$s where $1\leq j\leq \mc{L}$ to be the design matrix of Price and Scarlett \cite{wang2023quickly,price2020fast}. After the circling and zeroing process, with high probability, we will have $\leq k$ circles in each $B_j$. Hence applying LDs that can decode up to only $k$ defectives will work. By Lemma~\ref{lem:const}, we can choose $\zeta=16$. Also, by the analysis of Error 2 and the decoding algorithm, false positives will appear with probability $\mc{O}(N^{-1})$, while our algorithm will find all truly positives. \end{proof}

\subsection{Generalized Noiseless NEON}\label{sec:generalized}

Now we have established the BasicNEON scheme. The following are potential improvements to this scheme to be generalized into the Noiseless NEON. 

\textbf{Improvement 1:} BasicNEON works for the regime $K=o(N^{0.8})$. Can we generalize the sparse regime parameter $\alpha=0.8$ to a more general form? 

\textbf{Improvement 2:} Can we make the Error 1 probability to be significantly smaller than $\mc{O}(K^{-1})$, say $\mc{O}(K^{-b})$ for any given $b\geq 1$? As for small $K$, the Error 1 probability $p_1\approx \frac{\lambda}{K\log K}$ is not too small as desired. For example, in BasicNEON, $\lambda=C*\delta=200$, and in order for $p_1<0.01$, we need $K>4727$. 

\textbf{Improvement 3:} Can we improve the constant $\zeta$ to be less than $16$?

Next, we deduce the general Noiseless NEON scheme by analyzing these potential three improvements. 

\emph{Regarding Improvement 1,} recall that the sparse regime parameter $\alpha=0.8$ comes from Lemma~\ref{lem:basicneon}, where we had $\mc{O}(N^{-10(1-\alpha))})\ll\mc{O}(N^{-2})$. This is essentially $10(1-\alpha)>2$ which is equivalent to $\alpha<0.8$. 

More generally, we can replace the number $10$ with $C$, and we can also replace the number $2$ with any number greater than one, say $1+\eta$ where $\eta>0$ is given. Then we have 
$$C(1-\alpha)>1+\eta\Rightarrow\alpha<1-\frac{1+\eta}{C}.$$
By a similar analysis of Error 2 in Section~\ref{sec:errorprob} we would get $p_2=\mc{O}(N^{-\eta})$ instead of $\mc{O}(N^{-1})$. 

\emph{Regarding Improvement 2,} recall in Section~\ref{sec:errorprob} that the Error 1 probability $p_1=\mc{O}( \frac{\lambda}{K\log K})$ originates from $p:=K^{-2}$, and the reason for $p=K^{-2}$ originates from $\frac{e}{\delta}=\frac{e}{e^{3}}=e^{-2}$.

If we choose a bigger constant $\delta$, say $\delta=e^{b+2}$ for some positive integer $b$, then we would get $p=K^{-b-1}$, and then $p_1= \mc{O}(\frac{\lambda}{K^{b}\log K})=\mc{O}(K^{-b})$.

\emph{Regarding Improvement 3,} notice that the improvement of Wang, Gabrys and Guruswami \cite{wang2023quickly} from the original Price and Scarlett \cite{price2020fast} is with respect to the constant $\zeta$. We have the following lemma: 

\begin{lemma}\label{lem:newconst} Let $\gamma$ be defined as in (ii), and choose the constant $\zeta:=\log(2-4\epsilon)^{-2}$, where $\epsilon>0$ is small, then $\gamma<\frac{k}{2\epsilon N}$. 
\end{lemma}
\begin{proof} The proof will be completely similar to the proof of Lemma~\ref{lem:const}, with the recursive definition of $F_l'(1)$ changed to 
\[F_{l+1}'(1)=(1-2\epsilon)F_l'(1)+k \tag{a}\]
The new recursive definition above comes from the new recursive relation 
\[F_{l+1}(q)=F_l\left((\frac{1}{2}+\epsilon)+(\frac{1}{2}-\epsilon) q^2\right)\cdot q^k \tag{b}\]
The equation (b) comes from changing the equation (1) in Section III, E of \cite{wang2023quickly} with the new constant $\zeta$. 
By taking the derivative of (b) and plugging in $q=1$, we get (a). 
\end{proof}

Notice that, by changing the constant $\zeta$, Lemma~\ref{lem:basicneon} would still hold true, as changing $\gamma$ from $\frac{2k}{N}$ to $\frac{k}{2\epsilon N}$ would not affect its proof. However, by Theorem 1 in \cite{wang2023quickly}, changing $\zeta$ from $16$ to $\log(2-4\epsilon)$ would induce an extra $\epsilon^{-2}$ term in the decoding complexity. 

We conclude the Noiseless NEON with the following theorem: 

\begin{theorem}\label{thm:noiselessneon} (Noiseless NEON) Suppose, $C,\epsilon,\eta>0$, and $b$ is a positive integer. Let $\zeta:=\log(2-4\epsilon)^{-2}$. For the regime $K=o(N^{1-\frac{2(1+\eta)}{C}})$, 
there exists a noiseless, exact-recovery probabilistic GT scheme with $\lambda\zeta K\log N=Ce^{b+2}\zeta K\log N$ tests and decoding complexity $\mc{O}(\epsilon^{-2}K\log N)$, with error probability $<p_1+p_2=\mc{O}(K^{-b}+N^{-\eta})$. 
\end{theorem}
\section{Noisy NEON scheme: $\rho$ - False Positive Channel}\label{sec:fpcneon}
Previously, we discussed the Noiseless NEON scheme. In this section, we consider the $\rho$-False Positive Channel ($\rho$-FPC), where each negative test is flipped to a positive test with a constant probability $\rho$. 

Our scheme will be exactly the same as that for Noiseless NEON, but due to the extra false positive flipping probability $\rho$, in the analysis we need to consider this, and we will give a mild restriction on this flipping probability $\rho$ in the main theorem of this section. 

In Section~\ref{sec:toyexample}, we will give a concrete example of our scheme under the $\rho$-FPC when $\rho=\frac{1}{32}$, and in Section~\ref{sec:rhofpcneon}, we will generalize it into the main theorem of this section. 

\subsection{Analysis of the case when $\zeta=32$ and $\rho=\frac{1}{32}$}\label{sec:toyexample} 

Recall that our decoding algorithm will follow a binary splitting tree locally for each block $B_j$. 

Hence, locally, the analysis of our decoding algorithm will follow that in \cite{wang2023quickly} (see Section III.D-F).
The only difference is that a non-defective item may be placed in a false positive test. In the noiseless case, this happens with probability zero, while in this case, it happens with probability $\rho$. 

On the other hand, a non-defective item is mistakenly placed in a truly positive test with probability at most $\frac{1}{\zeta}$ (In the proof of Theorem~\ref{thm:fpcneon}, we will call this probability as the \emph{misplacement probability}). Hence, a non-defective item is claimed non-defective with a probability at least 
$(1-\rho)\cdot (1-\frac{1}{\zeta})=(\frac{31}{32})^2>\frac{15}{16}$.

Having this fact, the analysis of Section III in \cite{wang2023quickly} will follow without any change, further leading the (global) analysis in Section~\ref{sec:noiselessneon} follow smoothly. In other words, our example $\zeta=32$ and $\rho=\frac{1}{32}$ can be regarded the same (in fact, weaker) as that in~\cite{wang2023quickly} where $\zeta=16$ and $\rho=0$. 

%Compared to Noiseless NEON, the only difference of toy FPC-NEON is that the number of tests doubled, as the constant $\zeta$ should be changed from $16$ to $32$. All other things stay the same.  
\subsection{General $\rho$-FPC NEON}\label{sec:rhofpcneon}
Having seen the example in Section~\ref{sec:toyexample}, we generalize it into the following theorem:
\begin{theorem}\label{thm:fpcneon} ($\rho$-FPC NEON) Suppose we have a $\rho$-False Positive Channel,  parameters $C,\zeta,\epsilon,\eta>0$, such that $\zeta \log(2-4\epsilon)^{2}\geq\zeta\rho-\rho+1$, and a positive integer $b$. For the regime $K=o(N^{1-\frac{2(1+\eta)}{C}})$, 
there exists an exact-recovery probabilistic GT scheme with $\lambda\zeta K\log N=Ce^{b+2}\zeta K\log N$ tests and decoding complexity $\mc{O}(\epsilon^{-2}K\log N)$, with error probability $<p_1+p_2=\mc{O}(K^{-b}+N^{-\eta})$.
\end{theorem}
\begin{proof}  
It remains to show the equivalence of this setting with that in Theorem~\ref{thm:noiselessneon}.  
Using the arguments in Section~\ref{sec:toyexample}, the misplacement probability, which is $\frac{1}{2}-\epsilon$ in~\cite{wang2023quickly}, should be 
$$1-(1-\rho)(1-\frac{1}{\zeta})=\frac{\zeta\rho-\rho+1}{\zeta}$$
under this setting. 
Hence, the constant $\zeta$ in Theorem~\ref{thm:noiselessneon} should be changed to the inverse of this misplacement probability, i.e. $\frac{\zeta}{\zeta\rho-\rho+1}$. Let $\zeta':=\frac{\zeta}{\zeta\rho-\rho+1}$, recall Lemma~\ref{lem:newconst}, we should require $\zeta'\geq \log(2-4\epsilon)^{-2}$, which is equivalent to $\zeta \log(2-4\epsilon)^{2}\geq\zeta\rho-\rho+1$. 

On the other hand, regarding the number of tests, the old constant $\zeta$ will be kept unchanged. \end{proof} 

\begin{example} For convenience, let's set our target $\zeta'$ to be $16$ instead of $\log(2-4\epsilon)^{-2}$. 

As demonstrated at the end of Section~\ref{sec:toyexample} and combined with Theorem~\ref{thm:fpcneon}, when $\zeta=32$ and $\rho=\frac{1}{32}$, we would get $\zeta'=\frac{32}{32\cdot\frac{1}{32}-\frac{1}{32}+1}$ which is barely bigger than $16$, so the FPC-NEON with $\rho=\frac{1}{32}$ would require about $32Ce^{b+2}K\log N$ tests. 
If we raise the flipping probability, say $\rho=\frac{1}{16}$, then solving $\frac{\zeta}{\zeta\rho-\rho+1}\geq 16$ would give no solution. However, if $\rho$ is slightly smaller than $\frac{1}{16}$, say $\rho=\frac{1}{20}$, then we get $\zeta\geq 76$, so we will need $76Ce^{b+2}K\log N$ tests. 
On the other hand, if $\rho$ is small, then $\zeta$ would be close to $16$. For example, take $\rho:=\frac{1}{256}$, then we can pick $\zeta=17$, leading to the number of tests to be only $17Ce^{b+2}K\log N$. 
\end{example}
\subsection{Summary of the General Noiseless and the $\rho$-FPC NEON Scheme}\label{sec:algo1}
Assume that every parameter satisfies that in Theorem~\ref{thm:noiselessneon}, Theorem~\ref{thm:fpcneon} and the description before.

\begin{algorithm}[htb]
\caption{Noiseless and $\rho$-FPC NEON: Test matrix}\label{alg:noiseless-test}
\begin{algorithmic}[]
    \Require Number of items $N$, upper bound of number of defective items $K$, local matrix $B$, and parameters $\lambda,C$. Let $\mc{L}:=\frac{\lambda K}{\log K}$.
    \State Initialize empty matrix $A$.
    \For{$i=1,2,...,\mc{L}$}
    \State $B_i:=B$
    \State Concatenate $B_i$ to $A$ vertically. 
    \EndFor
    \For{$i=1,2,...,N$}
    \State Randomly Circle $C$ symbols for column $i$, where a symbol is a column of $B_j$ for some $j$.
    \State Make all non-circled entries equal to zero. 
    \EndFor    
    \State \Return $A$
\end{algorithmic}
\end{algorithm}

\begin{algorithm}
\caption{Noiseless and $\rho$-FPC NEON: Decoding algorithm\label{alg:noiseless-decoding}}
\begin{algorithmic}[1]
    \Require Number of items $N$, upper bound of number of defective items $K$, the test matrix $A$, the test result vector $y$, and parameters $\lambda, C,\theta$. Let $\mc{L}:=\frac{\lambda K}{\log K}$.
    \State Partition $y=[y_1;y_2;...;y_{\mc{L}}]$ and $A=[A_1;A_2;...;A_{\mc{L}}]$.
    \State Initialize empty list $R$, which can keep track of the multiplicity of each element. 
    \For{$i=1,2,...,\mc{L}$}
    \State Apply the decoding algorithm of~\cite{wang2023quickly} on $y_i$ and $A_i$, and denote the decoding result to be $R_i$, which is a list consisting of potentially defective items. 
    \State Append $R_i$ to $R$, with multiplicity. 
    \EndFor
    \For{each element in $R$}
    \If{this element has multiplicity smaller than $C$}
    \State Delete this element from $R$. 
    \EndIf
    \EndFor
    \State \Return $R$
\end{algorithmic}
\end{algorithm}
\section{Noisy NEON scheme: $\rho$ - Binary Symmetric Channel}\label{sec:BSCneon}
In the previous sections, we presented the Noiseless NEON and $\rho$-FPC NEON, where there are no false negative noises. In this section, we consider the $\rho$-BSC setting and propose the $\rho$-BSC NEON. We will keep using the same notations, but our construction and analysis will be different, as we introduced the false negative channel. 

To distinguish between the false positive and false negative channels, we will keep using different notations $\rho$ and $\rho'$ in this section. For BSC, take $\rho=\rho'$. 

In the following discussion, we will use the terms ``node" and ``item" interchangeably when a tree node represents a single item. 

Also, we will assume the sublinear sparse regime $K=\Theta(N^\alpha)$ where $0<\alpha<1$, although our scheme can work with sparser regimes. 
\subsection{Preliminaries}\label{sec:prelimbsc}
Using our binary tree model (Section~\ref{sec:fbs}), define a (sub)chain of the tree to be a \emph{defective chain} if every node in this chain should be tested positive. 

The following diagram illustrates a small example where there are $N=16$ items, and items 3 and 14 are defective, giving two defective chains of length 5. 
\begin{center}
\includegraphics[width=17cm, height=5cm]{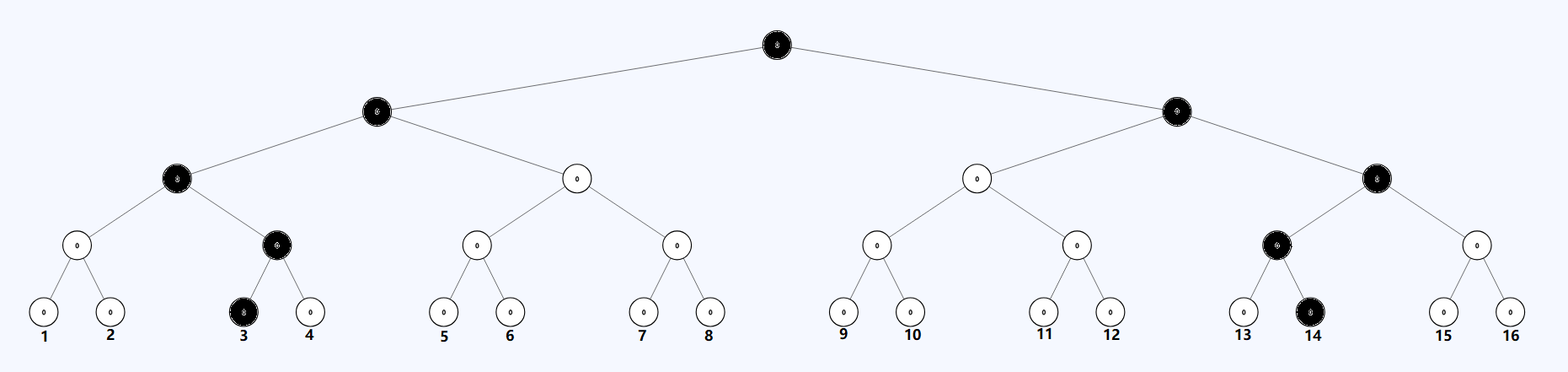}
\end{center}

In a false negative channel, some defective nodes can be falsely claimed to be non-defective, which will prevent the original version of the fast binary splitting decoding algorithm from finding such defective items, as stated in Section~\ref{sec:std}. Hence, in our decoding algorithm, we will examine the subtree of depth $r$ of some selected node, accordingly, compared to the original decoding algorithm, this will increase $D$ by a factor of approximately $2^r$. 

Using a similar idea as that in~\cite{price2023fast}, in our test design and decoding algorithm, we will use the binary tree model with $r$ extra levels. These extra levels do not have further branching and aim to eliminate all false positives. 
% Moreover, these extra levels will induce more tests, but the total number of tests $M$ will still be within the same asymptotic scale. 

The following diagram illustrates the extra levels of our binary tree model: 
\begin{center}
\includegraphics[height=6cm,width=11cm]{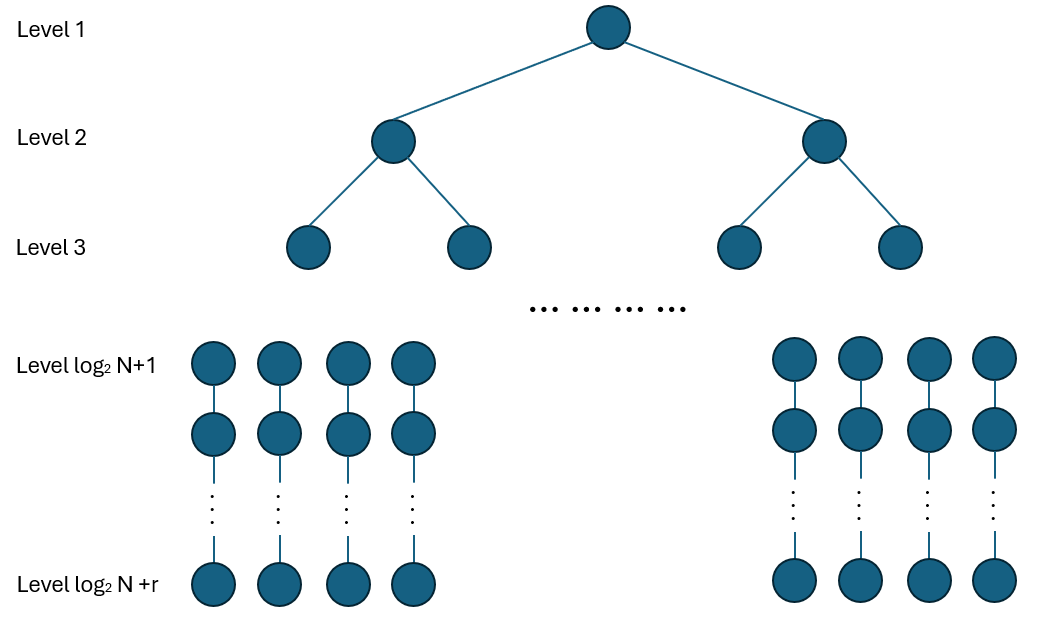}
\end{center}

We will call the level $\log_2 N+r$ as the \emph{ultimate level}, while calling the level $\log_2 N+1$ as the \emph{final level} in our analysis. 

\subsection{Testing procedure}\label{sec:testingprocedure}
Inspired by~\cite{price2023fast}, we will repeat the procedure of random assignment of nodes in $\zeta K$ tests in each level of the binary tree by $C'$ times, here $C'>0$ is a constant. Note that our $C'$ means the same as the parameter $N$ in~\cite{price2023fast}. 

For example, fix a level of our binary tree where there are $N':=9$ nodes. Label the nodes with numbers 1 to 9. Suppose $\zeta K=3$, then a typical realization of the $\zeta K$ tests using the original binary splitting algorithm will be $[1,4,6],[2,7,9],[3,5,8]$, meaning that the first test contains nodes 1,4,6; the second test contains nodes 2,7,9; and the third test contains nodes 3,5,8. 

In our algorithm, we will repeat this procedure by $C'$ times. Suppose $C'=3$, then a typical realization of the $C'\zeta K$ tests will be $[1,5,7],[2,3,9],[4,6,8];[1,2,9],[3,5,8],[4,6,7];[1,3,5],[2,6,8],[4,7,9]$. 

In the original binary splitting algorithm, we claim that a node is defective/non-defective only by looking at the result of the single test in which this node is present, so there is a probability of (at most) $\rho+\frac{1}{\zeta}$/$\rho'$ of this node being false positive/false negative. 

In our binary splitting algorithm, every node participates in $C'$ tests, We will claim this node to be defective only if more than $\frac{C'}{2}$ of these tests are positive. Otherwise, we claim this node to be non-defective. 

Consider the random variable $X':=X'_1+X'_2+...+X'_{C'}$ where each $X'_i\sim$ Ber($\rho'$). 

We have $\mathbb{E}[X']=C'\rho'$, and by Chernoff bound, we obtain
\[\mathbb{P}\left(X'>\frac{C'}{2}\right)=\mathbb{P}\left(X'>\frac{\mathbb{E}[X']}{2\rho'}\right)\leq \left(\frac{e}{\frac{1}{2\rho'}}\right)^{\frac{\mathbb{E}[X']}{2\rho'}}=(2e\rho')^{\frac{C'}{2}}\]

This means that a false negative node is claimed with probability less than $(2e\rho')^{\frac{C'}{2}}$. 

Similarly, a false positive node is claimed with probability less than $(2e(\rho+\frac{1}{\zeta}))^{\frac{C'}{2}}$. 

When $\rho=\rho'$, for the probability above to be less than one, we require
\[2e(\rho+\frac{1}{\zeta})<1.\]
\subsection{Decoding algorithm description and analysis}\label{sec:rhobscdecoding}
Without loss of generality, suppose $K$ is a power of $2$. During the decoding process, we will keep track of the \emph{multiplicity} of a node, which is defined to be the number of defective nodes in its corresponding chain. Also, we define and keep track of the \emph{density} of a node to be its multiplicity divided by its position (depth) in the chain. 

For example, consider the chain $1\rightarrow 1\rightarrow 1\rightarrow 0\rightarrow 0\rightarrow 1\rightarrow 1\rightarrow 1\rightarrow 1\rightarrow 0$, where $1$ represents a defective node and $0$ represents a non-defective node. The last node has a multiplicity $7$ and a density $0.7$.  

Let $p':=(2e\rho')^{\frac{C'}{2}}$ be the lower bound of $\mathbb{P}(X'>\frac{C'}{2})$ (see Section~\ref{sec:testingprocedure}), and let $r:=\epsilon\log_2 K$. Consider the random variable $Z\sim $Bin($r,p'$), which represents the number of false negatives of a given defective chain of length $r$. 

By Chernoff bound, the probability of a given defective chain having more than half of false negative nodes, is bounded by \[\mathbb{P}\left(Z>\frac{r}{2}\right)=\mathbb{P}\left(Z>\frac{1}{2p'}\mathbb{E}[Z]\right)\leq \left(\frac{e}{\frac{1}{2p'}}\right)^{\frac{\mathbb{E}[Z]}{2p'}}=(2p'e)^{\frac{r}{2}}=K^{\frac{\epsilon C'}{4}(\log_2\rho'+\log_2(2e)\cdot\frac{C'+2}{C'})}\] 

For the probability above to be $\ll K^{-\omega}$, for a given parameter $\omega \geq 1$, it is sufficient that \begin{equation}
    \frac{\epsilon C'}{4}\left(\log_2\rho'+\log_2(2e)\cdot\frac{C'+2}{C'}\right)<-\omega 
\end{equation}

Similarly, by a union bound, the probability of the existence of a false positive path of depth $r$ of a given negative node is upper bounded by 
\[K^{\frac{\epsilon C'}{4}(\log_2(\rho+\frac{1}{\zeta})+\log_2(2e)\cdot\frac{C'+2}{C'})}\cdot 2^{r}=K^{\frac{\epsilon C'}{4}(\log_2(\rho+\frac{1}{\zeta})+\log_2(2e)\cdot\frac{C'+2}{C'})+\epsilon}\]

For the probability above to be $\ll K^{-\omega}$ for a given parameter $\omega\geq 1$, it is sufficient that 

\begin{equation}
    \frac{\epsilon C'}{4}\left(\log_2(\rho+\frac{1}{\zeta})+\log_2(2e)\cdot\frac{C'+2}{C'}\right)+\epsilon<-\omega 
\end{equation}

The equations (1) and (2) imply that when $\rho$ and $\rho'$ are constants, $C'=\mc{O}(\epsilon^{-1})$.

\emph{The description of our decoding algorithm is as follows:} Step one of our decoding algorithm starts at the level $\log_2{K}-1$ of the binary tree, with $K$ nodes in this level. Then, we will examine all nodes in the sub-trees of depth $r$ of these $K$ nodes, which has complexity $\mc{O}(C'K\cdot 2^{r})=\mc{O}(\epsilon^{-1}K^{1+\epsilon})$. At depth $r+\log_2 K-1$, we will only keep those nodes with a density greater than $0.5$, and we call these nodes \emph{possibly defective}. 

Consider the following example, where a green node is a defective node found by the testing process. For some of the nodes, their multiplicity/density is written on their upper left/right in color black/red. Our decoding algorithm will find the defectives to be the nodes with a check mark, as their densities are greater than $0.5$. 
\begin{center}
\includegraphics[height=7.5cm,width=16cm]{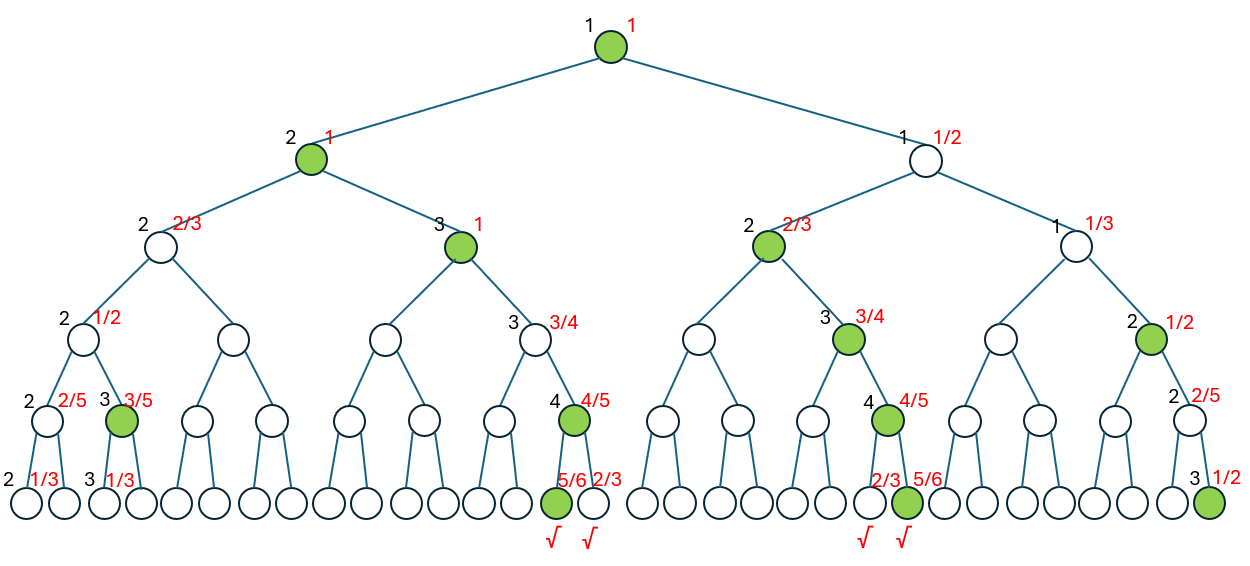}
\end{center}

Next, we will start at all these possibly defective nodes, set the densities and multiplicities to be one, and repeat the step one process, until we reach the final level. Without loss of generality, we will assume that $\log_2(\frac{N}{K})$ divides $r=\epsilon \log_2 K$. 

Finally, we will start at all the possibly defective nodes (items) in the final level and examine the results of nodes in their corresponding chains up to the ultimate level. If more than half of the nodes in the chain are defective, then this node (item) is claimed defective. Otherwise, it is claimed non-defective. 

\emph{Now, we analyze the performance of our decoding algorithm.} At the end of step one, we expect $\mc{O}(K^{1+\epsilon})$ possibly defective nodes. Among them, $\mc{O}(K^{1+\epsilon})$ are false positives and $\mc{O}(K)$ are truly positives. 

On one hand, we want to eliminate all the false positive nodes at the end of step one (and also all false positive nodes at the end of all subsequent steps). By Equation (2), this can be done by setting the parameter $\omega$ to be large enough, so that the total number of false positive nodes collected from all steps is smaller than $K^{\omega}$, and the elimination criterion of false positives at the end of step one is given by reading the densities of all leaf nodes at the end of step two, i.e. if all such leaf nodes have density $<0.5$, then the false positive node is eliminated. 

On the other hand, we want to keep all the truly positive nodes. By Equation (1) and the reasoning in the previous paragraph, this can also be done by setting the parameter $\omega$ to be large enough. By continuing our decoding algorithm, these $\mc{O}(K)$ truly positive nodes will be kept, while an extra $\mc{O}(K^{1+\epsilon})$ false positive nodes are yielded. 

The analysis will follow similarly throughout all steps. At the beginning of the final step, we will have $\mc{O}(K^{1+\epsilon})$ false positive nodes and $\mc{O}(K)$ truly positive nodes. By similar reasoning as before, although we are dealing with chains of length $r$ rather than subtrees of depth $r$, all the false positive nodes will be eliminated and all the truly positive nodes will be kept with high probability. 

The following diagram illustrates the process of our decoding algorithm and some remarks from the analysis:
\begin{center}
\includegraphics[height=8cm, width=14cm]{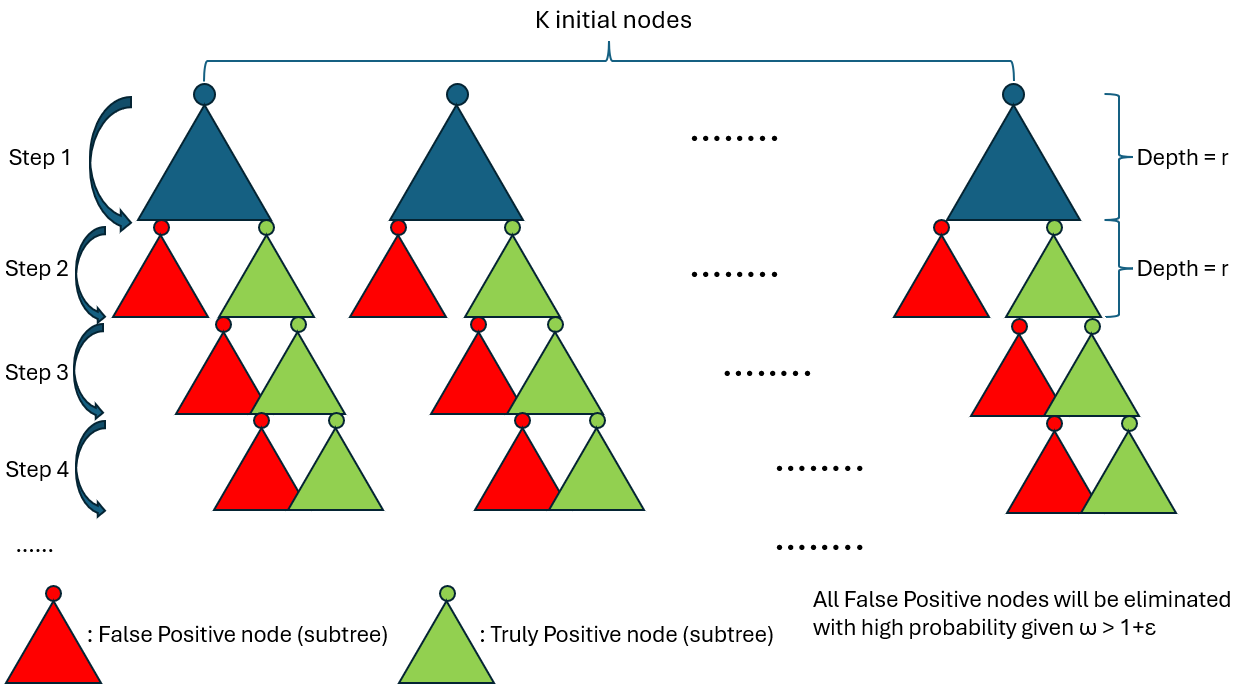}
\includegraphics[height=6cm, width=10cm]{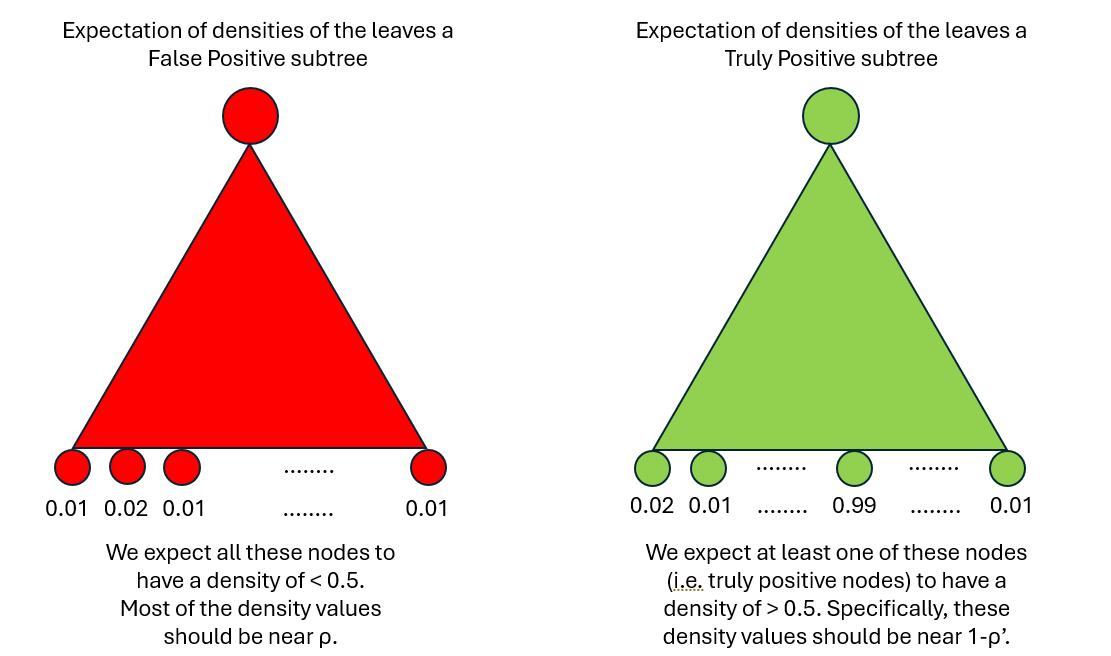}
\end{center}

Note that the number of steps will be \[\frac{\log_2 \frac{N}{K}}{r}+1=\frac{\log_2 \frac{N}{K}}{\epsilon\log_2 K}+1\]
which will be a constant $\leq \frac{\alpha}{\epsilon}+1 $ in the sublinear regime $K=\Theta(N^\alpha)$.

At the end of each step except the final step, under mild conditions, with high probability there will be $\mc{O}(K)$ truly positive nodes and $\mc{O}(\epsilon^{-1}K^{1+\epsilon})$ false positive nodes. Each node induces $2^r=K^\epsilon$ extra nodes to be examined. Hence, altogether there will be $$\mc{O}((\epsilon^{-1}K^{1+\epsilon})\cdot (\frac{\log_2 \frac{N}{K}}{\epsilon\log_2 K}+1)\cdot 2^r)=\mc{O}(\epsilon^{-2}K^{1+2\epsilon})$$ 
nodes to be examined. 
Substitute $2\epsilon$ by $\epsilon$, the above equals $\mc{O}(\epsilon^{-2}K^{1+\epsilon})$. 

Combined with equations (1) and (2), we want to choose the parameter $\omega$ such that $$\epsilon^{-2}K^{1+\epsilon}<K^\omega$$  

Hence, as $K\rightarrow \infty$, choosing any $\omega>1+\epsilon$ would suffice. 

% It follows that, in the regime $K=\omega(\log N)$, choosing $\omega>1+\epsilon$ will be sufficient. Furthermore, in the regime $K=\Omega((\log N)^u)$, choosing $\omega=\frac{1}{u}+1+\epsilon$ will be sufficient. In general, we need $K=\Omega(\text{polylog} N)$. 

\subsection{Summary of the $\rho$-BSC NEON}\label{sec:summaryrhobscneon}
\emph{Number of tests }($M$): There will be $(\log_2 N+r)-(\log_2 K+1)+1=\log_2(\frac{N}{K})+r$ levels, where $r=\epsilon\log_2 K$, and each level has $\zeta C'K$ tests, and note that $C'=\mc{O}(\epsilon^{-1})$ by equations (1) and (2), so \[M=\zeta C'K(\log_2 (\frac{N}{K})+\epsilon \log_2 K)=\mc{O}(\epsilon^{-1}K\log N).\]
\emph{Decoding complexity }($D$): In step one of our decoding, we will examine $K\cdot 2^r=K\cdot K^\epsilon=K^{1+\epsilon}$ nodes, each by $C'=\mc{O}(\epsilon^{-1})$ times. 

In each of the following steps, we will examine at most $K^{1+\epsilon}\cdot C' =K^{1+2\epsilon}$ nodes, since each of the $K^{1+\epsilon}$ nodes corresponds to $2^r=K^\epsilon$ children to be examined. 

There are $\frac{\log_2 N}{\epsilon\log_2 K}+1=\frac{\alpha}{\epsilon}+1$ steps in total, so
\[D\leq \left(\frac{\log_2 N}{\epsilon\log_2 K}+1\right)\cdot K^{1+2\epsilon}\cdot \mc{O}(\epsilon^{-1})=\mc{O}(\epsilon^{-2}K^{1+2\epsilon})=\mc{O}(\epsilon^{-2}K^{1+\epsilon})\]
% When $K=\Theta(N^\alpha)$, we will get $D=\mc{O}(\epsilon^{-1}K^{1+\epsilon})$ since $\frac{\log_2 N}{\log_2 K}=\alpha$ is a constant, which is better than~\cite{price2023fast} by a factor of $(\log N)^{1+\epsilon}$. 
\emph{Error probability}: There are two types of error in our scheme which was stated in the previous section. Here we restate them and calculate them in detail: 

\emph{The first type of error} is that there exists a false positive node not being eliminated. There are at most $\mc{O}(K^\omega)$ false positive nodes, where $\omega>1+\epsilon$, each with failure probability at most $K^{\frac{\epsilon C'}{4}(\log_2(\rho+\frac{1}{\zeta})+\log_2(2e)\cdot\frac{C'+2}{C'})+\epsilon}$. 

By union bound, the type one error probability is bounded by $\mc{O}(K^{\nu})$ where $\nu=\frac{\epsilon C'}{4}(\log_2(\rho+\frac{1}{\zeta})+\log_2(2e)\cdot\frac{C'+2}{C'})+\epsilon+\omega$. Choosing appropriate parameters $\rho,\zeta,\epsilon,C'$ will make $\nu$ reasonably small. 

\emph{The second type of error} is that there exists a truly positive node not being kept (i.e. False Negative). There are at most $\frac{\log_2 N}{\epsilon\log_2 K}\cdot K$ truly positive nodes, each with failure probability at most $K^{\frac{\epsilon C'}{4}(\log_2\rho'+\log_2(2e)\cdot\frac{C'+2}{C'})}$.

By comparison, the type two error will be smaller than the type one error. 

In summary, we have the following main theorem of this section: 

\begin{theorem}\label{thm:maintheoremBSC} Suppose we have a Binary Symmetric Channel with flipping probability $\rho$, parameters $C',\zeta,\epsilon>0$ and $0<\alpha<1$. Suppose $2e(\rho+\frac{1}{\zeta})<1$ and $ \frac{\epsilon C'}{4}\left(\log_2(\rho+\frac{1}{\zeta})+\log_2(2e)\cdot\frac{C'+2}{C'}\right)<-(1+2\epsilon)$. For the regime $K=\Theta(N^\alpha)$, there exists an exact-recovery probabilistic GT scheme with $\zeta C'K(\log_2 (\frac{N}{K})+\epsilon \log_2 K)=\mc{O}(\epsilon^{-1}K\log N)$ tests and decoding complexity $\mc{O}(\epsilon^{-2}K^{1+\epsilon})$, with error probability $\mc{O}(K^{\nu})$, where $\nu=\frac{\epsilon C'}{4}(\log_2(\rho+\frac{1}{\zeta})+\log_2(2e)\cdot\frac{C'+2}{C'})+1+2\epsilon$.\end{theorem}

\section{Noisy NEON scheme: $(\rho,\rho')$ - Binary Asymmetric Channel}\label{sec:BACneon} 
In this section, we consider the $(\rho,\rho')$ - BAC setting, where the false positive flipping probability $\rho$ is a constant while the false negative flipping probability $\rho'$ is constrained by $\mc{O}(K^{-\epsilon})$. 

\emph{Let's assume $\rho'=\beta K^{-\epsilon}$ in the following analysis. }
\subsection{Decoding procedure and analysis}
Our scheme in this case will essentially be the same as that in the previous section, so we will use the same notations. The only differences are on the false negative flipping probability and the selection of $r$. Here, we choose $r$ to be a constant. These will result in slight changes in the analysis.  

As in the previous section, we will assume $\log_2(\frac{N}{K})$ divides $r$. However, different from before, in the construction of extra levels of this scheme, we will have $C''\log K$ levels, where $C''$ is a constant, instead of just $r$ levels, to eliminate all false positive nodes. In other words, the ultimate level of our tree will be $\log_2 N+C''\log K$, not $\log_2 N+r$. 

\emph{Here we begin our analysis.} First, we analyze the key probability $\mathbb{P}(Z>\frac{r}{2})$ as appeared in Section~\ref{sec:rhobscdecoding}, which denotes the probability of a given truly positive node being wrongly eliminated:
\[\mathbb{P}(Z>\frac{r}{2})\leq (2p'e)^{\frac{r}{2}}=(2e)^{\frac{r}{2}}\cdot (2e\beta)^{\frac{rC'}{4}}K^{-\frac{C'\epsilon r}{4}}\]
This will be smaller than $K^{-\omega}$ as $K\rightarrow\infty$ when
\begin{equation}
    C'\epsilon r>4\omega
\end{equation} 

There will be $\frac{\log_2\frac{N}{K}}{r}+1$ steps in total, at the end of each step there will be $\mc{O}(K)$ truly positive nodes, so the total number of truly positive nodes will be $\leq K(\frac{\log_2\frac{N}{K}}{r}+1)=\mc{O}(K^{1+\mu})$ for any $\mu>0$ (Recall that we are in the sublinear regime $K=\mc{O}(N^{\alpha})$). 

Let $\omega=1+\mu$, then the type one error probability (see Section~\ref{sec:summaryrhobscneon}) will be upper bounded by $\mc{O}(K^\nu)$ for $\nu:=-\frac{C'\epsilon r}{4}+\omega$. 

Note that we do not need to analyze the final step in particular, in which there are chains of length $C''\log K$, and $C''\log K\gg r$. 

Next, we analyze the evolution of false positives, which will be slightly different from the analysis in Section~\ref{sec:BSCneon}, as the number of steps $\frac{\log_2\frac{N}{K}}{r}+1$ is not a constant. 

Consider the probability of a given false positive node not being eliminated, i.e. its corresponding chain has more than half of defective nodes. By a similar analysis as the quantity $\mathbb{P}(Z>\frac{r}{2})$, this probability is less than
\[p_0:=(2e(\rho+\frac{1}{\zeta})^{\frac{C'}{2}})^{\frac{r}{2}}=(2e)^{\frac{r}{2}}\cdot (\rho+\frac{1}{\zeta})^{\frac{C'r}{4}}\]

Following the same decoding algorithm as in Section~\ref{sec:BSCneon}, define the number of false positive nodes at the end of the step $l$ as $N_l$. Note that the same notation was used in Section~\ref{sec:fbs} as well, but here it means differently.  

At the beginning of the decoding algorithm, we start at the $K$ nodes at the level $\log_2 K-1$ of the tree, so we suppose $N_0=K$. Also, as in~\cite{wang2023quickly}, we define the probability generating function $F_l(q)$ of $N_l$ by $F_l(q)=\sum_{j=0}^{\infty} \mathbb{P}(N_l=j)\cdot q^j$.  

By our algorithm, each positive node from a given step will generate a certain number of false positives for the next step. Suppose that this number is $n_f$. Since there are $2^r$ leaf nodes, each with probability $p_0$ to be a false positive. Also, the false positivity of a given leaf node is nearly independent of the other nodes, so we have $\mathbb{P}(n_f=0)=\Theta( (1-p_0)^{2^r})$, $\mathbb{P}(n_f=1)=\Theta( 2^r\cdot p_0(1-p_0)^{2^r-1})$, and in general $\mathbb{P}(n_f=j)=\Theta({2^r\choose j} p_0^{j}(1-p_0)^{2^r-j})$. 

Define $f(q):=((1-p_0)+p_0q)^{2^r}$, and $g(q):=\sum_{j=0}^{2^r}\mathbb{P}(n_f=j)q^j$. By the argument above, $f(q)=\Theta(g(q))$. In the following analysis, we will assume $f(q)=g(q)$ without loss of generality. 

We then have $F_0(q)=q^K$, $F_1(q)=g(q)^K=f(q)^K$, $F_2(q)=f(f(q))^K$, and in general $F_l(q)=f^{(l)}(q)^K$, where $f^{\circ l}(q)$ denotes the function that composes $f(q)$ with itself for $l$ times. 

\begin{lemma}
Suppose $(1+p_0)^{2^r}<2$, then $\mathbb{P}(N_l\geq 2K)\leq 2^{-K}$ for any $l\geq 1$. 
\end{lemma}
\begin{proof}
We first show that $F_l(2)\leq 2^K$ for any $l\geq 1$. 
We have 
\[F_l(2)\leq 2^K \Leftrightarrow f^{(l)}(2)^K\leq 2^K \Leftrightarrow f^{(l)}(2)\leq 2. \]

Note that $f(1)=1$, $f(2)=(1+p_0)^{2^r}<2$, $f'(x)>0$, and $f''(x)>0$ on the interval $[1,2]$, so \[f^{(l)}(2)\leq f^{(l-1)}(2)\leq ... \leq f(2)<2.\]

Now we have $F_l(2)\leq 2^K$ for any $l\geq 1$, by the Chernoff bound, we have 
\[\mathbb{P}(N_l\geq 2K)\leq \frac{F_l(2)}{2^{2K}}=2^{-K}.\]
\end{proof}
By the lemma above, at the beginning of the final step, with probability $\geq 1-2^{-K}$, we will have less than $2K$ false positive items. The final step aims to eliminate all of these false positives. For each of them, we will check a chain of length $C''\log K$ and claim it to be positive only if more than half of the nodes in this chain are positive. By Chernoff bound, this happens with probability less than 
\[(2ep_0)^{\frac{C''\log K}{2}}=K^{\frac{C''}{2}\log(2ep_0)}\]

Given the total number of false positives we start is at most $2K=\mc{O}(K)$, we need 
\[\frac{C''}{2}\log(2ep_0)<-1\]

By a union bound, the type two error probability (see Section~\ref{sec:summaryrhobscneon}) is upper bounded by $K^{\frac{C''}{2}\log(2ep_0)+1}$ (Note that the error probability of having $>2K$ false positives at the beginning of final step, which is smaller than $2^{-K}$, is ignored, since it is much smaller than $K^{\frac{C''}{2}\log(2ep_0)+1}$). 

Regarding the decoding complexity, the analysis will be similar to that in Section~\ref{sec:summaryrhobscneon}, but the depth $r$ is a constant rather than $\mc{O}(\log K)$, resulting in a decoding complexity of $\mc{O}(\epsilon^{-1}K\log N)$ instead of $\mc{O}(\epsilon^{-2}K^{1+\epsilon})$.

% At the end of step one, we will get at most $2^r K$ false positive nodes and $\mc{O}(K)$ truly positive nodes. Continuing, at the end of Step two, we expect less than $(2e)^{\frac{r}{2}}\cdot \rho^{\frac{C'r}{2}}\cdot 2^r K+2^r K=(\sqrt{8e}\rho^{\frac{C'}{2}})^r K+2^rK$ false positive nodes. The addend $2^rK$ is produced by the $\mc{O}(K)$ truly positive nodes from the end of Step one. Continuing, at the beginning of the final step, we will expect less than $\mc{O}(2^rK)$ false positive nodes and $\mc{O}(K)$ truly positive nodes, 

In summary, we have the following main theorem of this section:
\begin{theorem}\label{thm:maintheoremBAC}
Suppose we have a Binary Asymmetric Channel with constant false positive flipping probability $\rho$ and false negative flipping probability $\rho'$, parameters $C',C'',\zeta,\epsilon,r>0$ and $0<\alpha<1$. Suppose $\rho'=\mc{O}(K^{-\epsilon})$, $2e(\rho+\frac{1}{\zeta})<1$, $(1+(2e)^{\frac{r}{2}}\cdot (\rho+\frac{1}{\zeta})^{\frac{C'r}{4}})^{2^r}<2$, $C'\epsilon r>4$, and $\frac{C''}{2}\log(2ep_0)+1<0$. For the regime $K=\Theta(N^\alpha)$, 
there exists an exact-recovery probabilistic GT scheme with $\zeta C'K(\log_2 (\frac{N}{K})+C''\log K)=\mc{O}(\epsilon^{-1}K\log N)$ tests and decoding complexity $\mc{O}(\epsilon^{-1}K\log N)$, with error probability $\mc{O}(K^{\max\{\nu_1,\nu_2\}})$, where $\nu_1=1-\frac{C'\epsilon r}{4}$ and $\nu_2=\frac{C''}{2}\log(2ep_0)+1$. 
\end{theorem}
\subsection{Summary of the $\rho$-BSC NEON Scheme and the $(\rho,\rho')$-BAC NEON Scheme}\label{sec:algobsc}
Assume that every parameter satisfies that in Theorem~\ref{thm:maintheoremBSC}, Theorem~\ref{thm:maintheoremBAC} and the description before.
\begin{algorithm}
\caption{$\rho$-BSC and $(\rho,\rho')$-BAC NEON: Test matrix}\label{alg:noisy-test}
\begin{algorithmic}[]
    \Require Number of items $N$, upper bound of number of defective items $K$, and parameters $C',C'',\zeta,\epsilon,r$. 
    \State Initialize empty matrix $A$ of dimension $M$ by $N$, where $M:=\zeta C'K(\log_2 (\frac{N}{K})+r')$, such that $r'=\epsilon \log_2 K$ for $\rho$-BSC while $r'=C''\log K$ for $(\rho,\rho')$-BAC.
    \For{$i=\log_2 K+1,\log_2 K+2,...,\log_2 N$}
    \For{$j=1,2,...,C'$}
    \For{$l=1,2,...,2^{i-1}$}
    \State Place all items in the $l$-th node into one of the tests $\{u+1,u+2,...,u+\zeta K\}$ uniformly random, 
    \State where $u:=(i-\log_2K-1)\zeta C'K+(j-1)\zeta K$.
    \EndFor
    \EndFor
    \EndFor
    \For{$i=\log_2N+1,\log_2 N+2,...,\log_2N+r'$}
    \For{$j=1,2,...,C'$}
    \For{$l=1,2,...,N$}
    \State Place the $l$-th item into one of the tests $\{u+1,u+2,...,u+\zeta K\}$ uniformly random, 
    \State where $u:=(i-\log_2K-1)\zeta C'K+(j-1)\zeta K$.
    \EndFor
    \EndFor
    \EndFor    
    \State \Return $A$
\end{algorithmic}
\end{algorithm}

\begin{algorithm}[!]
\caption{$\rho$-BSC and $(\rho,\rho')$-BAC NEON: Decoding algorithm}\label{alg:noisy-decoding}
\begin{algorithmic}[1]
    \Require Number of items $N$, upper bound of number of defective items $K$, and parameters $C',C'',\zeta,\epsilon,r$, and outcomes of the tests. 
    \For{$i=1,2,...,K$}
    \State Initialize $S_i:=\{\mc{N}_i\}$, where $\mc{N}_i$ denotes the $i$-th node at level $\log_2K+1$ of the binary tree. 
    \For{$step=1,2,...,\frac{\log_2(\frac{N}{K})}{r}$}
    \State Initialize $S_i'=\varnothing$.
    \For{each node $s\in S_i$}
    \State Do depth first search of the subtree of depth $r$ of $s$. 
    \State For each node, update its density (See Section~\ref{sec:rhobscdecoding}). 
    \State For the $2^r$ leaf nodes, if it has a density $>0.5$, then include this node in $S_i'$. 
    \EndFor
    \State $S_i:=S_i'$. 
    \EndFor
    \EndFor
    \For{each node $s\in S_1\cup S_2$ $\cup$ ... $\cup$ $S_K$}
    \State Check the test results of nodes in the chain corresponding to $s$ (i.e. extra levels in Section~\ref{sec:prelimbsc}).
    \If{more than half of the nodes in the chain are positive}
    \State Include $s$ in $R$. 
    \EndIf
    \EndFor
    \State \Return $R$
\end{algorithmic}
\end{algorithm}
\section{Conclusions and Remarks}
In this paper, we first generalized the original fast binary splitting algorithm from the noiseless setting to the $\rho$-FPC setting, with the idea of local decoding, successfully achieving the optimal asymptotic bounds for both the number of tests $M$ and the decoding complexity $D$. Although applying the original binary splitting scheme with improved constant parameters will also achieve the same objective, our idea is new and can be potentially applied to other problems.  

In the second part, we modified Price-Scarlett-Tan's method so that the decoding complexity is improved by a factor of $\mc{O}((\log N)^{1+\epsilon})$ within the sublinear regime $K=\mc{O}(N^
\alpha)$ where $0<\alpha<1$.  The main idea here is to keep track of the density of the nodes so that we only need to visit each node once, instead of visiting them multiple times as in~\cite{price2023fast}. We also showed that in the case when the false negative flipping probability $\rho'$ is $\mc{O}(K^{-\epsilon})$, we can achieve the asymptotically optimal number of tests and decoding complexity at the same time. 

We conjecture that to find no false negatives, the decoding algorithm using the binary splitting method must examine sub-trees of depth $r=\mc{O}(\log K)$ of at least $K$ nodes, resulting in a decoding complexity at least $\mc{O}(K^{1+\epsilon})$, so our scheme already achieves the best asymptotic bounds of $M$ and $D$ in the binary-splitting based NAPGT scheme, within the regime $K=\Theta(N^\alpha)$ for some $0<\alpha<1$. In other words, it may be impossible to achieve both $M=\mc{O}(K\log N)$ and $D=\mc{O}(K\log N)$ for the $\rho$-BSC NAPGT based on the binary-splitting method. This direction is left to future research. 

\paragraph{Acknowledgement.}
The authors would like to thank Vasileios Nakos for pointing out an error, and a resultant wrong result in a previous version of this paper.

AM would like to thank the Simons Institute  program on Error-correcting codes where some related literature were discussed, and acknowledge NSF award 2217058 for support in research.

\bibliographystyle{alpha}
\bibliography{references}

\newcommand{\etalchar}[1]{$^{#1}$}
\begin{thebibliography}{BMTW84}

\bibitem[AJS19]{aldridge2019group}
Matthew Aldridge, Oliver Johnson, and Jonathan Scarlett.
\newblock Group testing: an information theory perspective.
\newblock {\em arXiv preprint arXiv:1902.06002}, 2019.

\bibitem[AS12]{atia2012boolean}
George~K Atia and Venkatesh Saligrama.
\newblock Boolean compressed sensing and noisy group testing.
\newblock {\em Information Theory, IEEE Transactions on}, 58(3):1880--1901, 2012.

\bibitem[BB23]{brust2023effective}
David Brust and Johannes~J Brust.
\newblock Effective matrix designs for covid-19 group testing.
\newblock {\em BMC bioinformatics}, 24(1):26, 2023.

\bibitem[BMTW84]{berger1984random}
Toby Berger, Nader Mehravari, Don Towsley, and Jack Wolf.
\newblock Random multiple-access communication and group testing.
\newblock {\em IEEE Transactions on Communications}, 32(7):769--779, 1984.

\bibitem[CJBJ17]{cai2017efficient}
Sheng Cai, Mohammad Jahangoshahi, Mayank Bakshi, and Sidharth Jaggi.
\newblock Efficient algorithms for noisy group testing.
\newblock {\em IEEE Transactions on Information Theory}, 63(4):2113--2136, 2017.

\bibitem[CN20]{cheraghchi2020combinatorial}
Mahdi Cheraghchi and Vasileios Nakos.
\newblock Combinatorial group testing and sparse recovery schemes with near-optimal decoding time.
\newblock In {\em 2020 IEEE 61st Annual Symposium on Foundations of Computer Science (FOCS)}, pages 1203--1213. IEEE, 2020.

\bibitem[CS16]{cao2016combinatorial}
Chang-chang Cao and Xiao Sun.
\newblock Combinatorial pooled sequencing: experiment design and decoding.
\newblock {\em Quantitative Biology}, 4:36--46, 2016.

\bibitem[Dor43]{dorfman1943detection}
Robert Dorfman.
\newblock The detection of defective members of large populations.
\newblock {\em The Annals of Mathematical Statistics}, 14(4):436--440, 1943.

\bibitem[FM21]{flodin2021probabilistic}
Larkin Flodin and Arya Mazumdar.
\newblock Probabilistic group testing with a linear number of tests.
\newblock In {\em 2021 IEEE International Symposium on Information Theory (ISIT)}, pages 1248--1253. IEEE, 2021.

\bibitem[GIS08]{gilbert2008group}
Anna~C Gilbert, Mark~A Iwen, and Martin~J Strauss.
\newblock Group testing and sparse signal recovery.
\newblock In {\em Signals, Systems and Computers, 2008 42nd Asilomar Conference on}, pages 1059--1063. IEEE, 2008.

\bibitem[GSTV07]{gilbert2007one}
Anna~C Gilbert, Martin~J Strauss, Joel~A Tropp, and Roman Vershynin.
\newblock One sketch for all: fast algorithms for compressed sensing.
\newblock In {\em Proceedings of the thirty-ninth annual ACM symposium on Theory of computing}, pages 237--246, 2007.

\bibitem[GW23]{guruswami2023noise}
Venkatesan Guruswami and Hsin-Po Wang.
\newblock Noise-resilient group testing with order-optimal tests and fast-and-reliable decoding.
\newblock {\em arXiv preprint arXiv:2311.08283}, 2023.

\bibitem[HD06]{hwang2006pooling}
Frank Kwang-ming Hwang and Ding-zhu Du.
\newblock {\em Pooling designs and nonadaptive group testing: important tools for DNA sequencing}, volume~18.
\newblock World Scientific, 2006.

\bibitem[HSP20]{hogan2020sample}
Catherine~A Hogan, Malaya~K Sahoo, and Benjamin~A Pinsky.
\newblock Sample pooling as a strategy to detect community transmission of sars-cov-2.
\newblock {\em Jama}, 323(19):1967--1969, 2020.

\bibitem[LCPR19]{lee2019saffron}
Kangwook Lee, Kabir Chandrasekher, Ramtin Pedarsani, and Kannan Ramchandran.
\newblock Saffron: A fast, efficient, and robust framework for group testing based on sparse-graph codes.
\newblock {\em IEEE Transactions on Signal Processing}, 67(17):4649--4664, 2019.

\bibitem[LG08]{luo2008neighbor}
Jun Luo and Dongning Guo.
\newblock Neighbor discovery in wireless ad hoc networks based on group testing.
\newblock In {\em 2008 46th Annual Allerton Conference on Communication, Control, and Computing}, pages 791--797. IEEE, 2008.

\bibitem[Maz16]{mazumdar2015nonadaptive}
Arya Mazumdar.
\newblock Nonadaptive group testing with random set of defectives.
\newblock {\em {IEEE} Trans. Information Theory}, 62(12):7522--7531, 2016.

\bibitem[MM24]{matsumoto2024robust}
Namiko Matsumoto and Arya Mazumdar.
\newblock Robust 1-bit compressed sensing with iterative hard thresholding.
\newblock In {\em Proceedings of the 2024 Annual ACM-SIAM Symposium on Discrete Algorithms (SODA)}, pages 2941--2979. SIAM, 2024.

\bibitem[MMP23]{matsumoto2023improved}
Namiko Matsumoto, Arya Mazumdar, and Soumyabrata Pal.
\newblock Improved support recovery in universal one-bit compressed sensing.
\newblock {\em IEEE Transactions on Information Theory}, 2023.

\bibitem[PS20]{price2020fast}
Eric Price and Jonathan Scarlett.
\newblock A fast binary splitting approach to non-adaptive group testing.
\newblock {\em arXiv preprint arXiv:2006.10268}, 2020.

\bibitem[PST23]{price2023fast}
Eric Price, Jonathan Scarlett, and Nelvin Tan.
\newblock Fast splitting algorithms for sparsity-constrained and noisy group testing.
\newblock {\em Information and Inference: A Journal of the IMA}, 12(2):1141--1171, 2023.

\bibitem[VFH{\etalchar{+}}21]{verdun2021group}
Claudio~M Verdun, Tim Fuchs, Pavol Harar, Dennis Elbr{\"a}chter, David~S Fischer, Julius Berner, Philipp Grohs, Fabian~J Theis, and Felix Krahmer.
\newblock Group testing for sars-cov-2 allows for up to 10-fold efficiency increase across realistic scenarios and testing strategies.
\newblock {\em Frontiers in Public Health}, 9:583377, 2021.

\bibitem[WGG23]{wang2023quickly}
Hsin-Po Wang, Ryan Gabrys, and Venkatesan Guruswami.
\newblock Quickly-decodable group testing with fewer tests: Price--scarlett’s nonadaptive splitting with explicit scalars.
\newblock In {\em 2023 IEEE International Symposium on Information Theory (ISIT)}, pages 1609--1614. IEEE, 2023.

\bibitem[YAT{\etalchar{+}}20]{yelin2020evaluation}
Idan Yelin, Noga Aharony, Einat~Shaer Tamar, Amir Argoetti, Esther Messer, Dina Berenbaum, Einat Shafran, Areen Kuzli, Nagham Gandali, Omer Shkedi, et~al.
\newblock Evaluation of covid-19 rt-qpcr test in multi sample pools.
\newblock {\em Clinical Infectious Diseases}, 71(16):2073--2078, 2020.

\end{thebibliography}

\end{document}